\documentclass[a4paper]{scrartcl}

\usepackage{amsmath,amssymb,amsthm,latexsym}
\usepackage{algorithm2e}
\usepackage{multicol}

\theoremstyle{plain}
\newtheorem{lemma}{Lemma}
\newtheorem{theorem}{Theorem}
\newtheorem{proposition}{Proposition}
\newtheorem{corollary}{Corollary}

\theoremstyle{definition}

\usepackage{booktabs}

\usepackage{tikz}

\newcommand{\s}{\overline{s}}
\newcommand{\cH}{{\mathcal{H}}}
\newcommand{\cP}{{\mathcal{P}}}
\newcommand{\cA}{{\mathcal{A}}}
\newcommand{\Q}{{\mathcal{Q}}}
\newcommand{\N}{{\mathbb{N}}}

\newcommand{\Oh}{\mathcal{O}}

\newcommand{\probname}[1]{\textsc{\lowercase{#1}}\xspace}
\newcommand{\WSP}[1]{\textsc{wsp}#1\xspace}

\newcommand{\FPT}{\ensuremath{\mathsf{FPT}}\xspace}
\newcommand{\NP}{\ensuremath{\mathsf{NP}}\xspace}

\begin{document}

\title{Polynomial Kernels and User Reductions for the Workflow Satisfiability Problem\thanks{An extended abstract appears in the proceedings of IPEC 2014.}}

\author{Gregory Gutin\footnote{Royal Holloway, University of London, UK, \texttt{$\{$G.Gutin$,$Magnus.Wahlstrom$\}$@rhul.ac.uk}} \and Stefan Kratsch\footnote{TU Berlin, Germany, \texttt{stefan.kratsch@tu-berlin.de}} \and Magnus Wahlstr{\"o}m\footnotemark[2]}

\maketitle

\begin{abstract}
The \probname{workflow satisfiability problem} (\WSP{}) is a problem of practical interest that arises whenever tasks need to be performed by authorized users, subject to constraints defined by business rules.  We are required to decide whether there exists a \emph{plan} -- an assignment of tasks to authorized users -- such that all constraints are satisfied.

The \WSP{} is, in fact, the \probname{conservative Constraint Satisfaction Problem} (i.e., for each variable, here called \emph{task}, we have a unary authorization constraint) and is, thus, \NP-complete. It was observed by Wang and Li (2010) that the number $k$ of tasks is often quite small and so can be used as a parameter, and several subsequent works have studied the parameterized complexity of \WSP{} regarding parameter~$k$.

We take a more detailed look at the kernelization complexity of \WSP{($\Gamma$)} when~$\Gamma$ denotes a finite or infinite set of allowed constraints. Our main result is a dichotomy for the case that all constraints in~$\Gamma$ are regular: (1) We are able to reduce the number~$n$ of users to~$n'\leq k$. This entails a kernelization to size poly$(k)$ for finite~$\Gamma$, and, under mild technical conditions, to size poly$(k+m)$ for infinite~$\Gamma$, where~$m$ denotes the number of constraints. (2) Already~\WSP{(\uppercase{$R$})} for some~$R\in\Gamma$ allows no polynomial kernelization in~$k+m$ unless the polynomial hierarchy collapses.
\end{abstract}

\section{Introduction} \label{sec:intro}

A business process is a collection of interrelated tasks that are performed by users in order to achieve some objective.
In many situations, a task can be performed only by certain \emph{authorized} users; formally, every task is accompanied by an \emph{authorization list} of all users who are authorized to perform the task. Additionally, either because of the particular requirements of the business logic or security requirements, we may require that certain sets of tasks cannot be performed by some sets of users~\cite{Cr05}. 
Such restrictions are referred to as \emph{constraints}, and may include rules such as \emph{separation-of-duty} (also known as the ``two-man'' rule), which may be used to prevent sensitive combinations of tasks being performed by a single user, and \emph{binding-of-duty}, which requires that a particular combination of tasks is performed by the same user.
The use of constraints in workflow management systems to enforce security policies has been studied extensively in the last fifteen years; see, e.g.,~\cite{BeFeAt99,Cr05,WangLi10}.

It is possible that the combination of constraints and authorization lists is ``unsatisfiable'', in the sense that there does not exist an assignment of users to tasks (called a \emph{plan}) such that all constraints are satisfied and every task is performed by an authorized user. A plan that satisfies all constraints and allocates an authorized user to each task is called \emph{valid}.
The \probname{workflow satisfiability problem} (\WSP{}) takes a workflow specification as input and returns a valid plan if one exists and \textsc{no} otherwise.
It is important to determine whether a business process is satisfiable or not, since an unsatisfiable one can never be completed without violating the security policy encoded by the constraints and authorization lists. 

Let us illustrate the above notions by the following simple instance $W^*$ of \WSP{}. Let the tasks be $s_1,s_2,s_3$, authorizations lists $A(s_1)=\{u_1,u_2,u_3\}$, $A(s_2)=\{u_1,u_4,u_5\}$, $A(s_3)=\{u_1,u_6\}$, a binding-of-duty constraint $s_1=s_2$  (meaning that $s_1$ and $s_2$ must be assigned the same user) and two separation-of-duty constraints $s_1\neq s_3$ and $s_2\neq s_3$. Note that the only valid plan is the assignment of $s_1$ and $s_2$ to $u_1$ and $s_3$ to $u_6.$

It is worth noting that \WSP{} is a special class of constraint satisfaction problems where for each variable $s$ (called a task in the \WSP{} language) we have an arbitrary unary constraint (called an authorization) that assigns possible values (called users) for $s$; this is called the conservative constraing satisfaction problem. Note, however, that while usually in constraint satisfactions problems the number of variables is much larger than the number of values,
for \WSP{} the number of tasks is usually much smaller than the number of users. It is important to remember that for \WSP{} we do not use the term 'constraint' for authorizations and so when we define special types of constraints, we do not extend these types to authorizations, which remain arbitrary.

Wang and Li  \cite{WangLi10} were the first to observe that the number $k$ of tasks is often quite small and so can be considered as a parameter. As a result, \WSP{} can be studied as a parameterized problem. Wang and Li  \cite{WangLi10}  proved that, in general, \WSP{} is $\mathsf{W[1]}$-hard, but \WSP{} is fixed-parameter tractable
(\FPT) if we consider 
some special types of practical constraints which include
separation-of-duty and binding-of-duty constraints. 

Crampton et al. \cite{CrGuYeJournal} found a faster fixed-parameter algorithm to solve the special cases of \WSP{} studied in \cite{WangLi10} and 
showed that the algorithm can be used for a wide family of constraints called regular (in fact, regular constraints include all constraints studied in  \cite{WangLi10}).
Subsequent research has demonstrated the existence of fixed-parameter algorithms for \WSP{} in the presence of other constraint types~\cite{CoCrGaGuJo13,CCGJR}. In particular, Cohen et al.~\cite{CoCrGaGuJo13} showed that \WSP{} with only so-called user-independent constraints is \FPT. A constraint $c$ on tasks $t_1,\ldots ,t_r$ is user-independent when for any tuples $u_{i_1},\ldots , u_{i_r}$ and $u_{j_1},\ldots , u_{j_r}$ of users such that $u_{i_p}=u_{i_q}$ if and only if $u_{j_p}=u_{j_q}$ for all $1\le p<q\le r$, $c$ is satisfied by the assignment of $t_{\ell}$ to $u_{i_{\ell}}$ for each $\ell \in [r]:=\{1,\ldots, r\}$ if and only if $c$ is satisfied by the assignment of $t_{\ell}$ to $u_{j_{\ell}}$ for each $\ell \in [r]$. Intuitively, given a satisfying assignment we may arbitrarily swap users for other users that are not presently assigned (and this may be iterated, giving arbitrary bijections). As an example, separation-of-duty and binding-of-duty are user-independent constraints.
Crampton et al.~\cite{CrGuYeJournal} also launched the study of polynomial and partially polynomial kernels (in the latter only the number of users is required to be bounded by a polynomial in $k$), 
but obtained results only for concrete families of constraints.

In this work, we explore the kernelization properties of \WSP{} in more detail. We focus on regular constraints, which are a special family of user-independent constraints, but since their definition is quite technical, we will defer it to  Section~\ref{sec:prels:wsp}.
We study both the possibility of polynomial kernels and of simplifying \WSP{} instances by reducing the set of users (i.e., partial kernels).%
\footnote{Such reductions are of interest by themselves as some practical \WSP{} algorithms iterate over users in search for a valid plan \cite{CoCrGaGuJo14}.} Our goal is to determine for which types of constraints such user-limiting reductions are possible, i.e., for which sets $\Gamma$ does the problem \WSP{($\Gamma$)} of \WSP{} restricted to using constraint types (i.e., relations) from $\Gamma$ admit a reduction to poly$(k)$ users? We study this question for both finite and infinite sets $\Gamma$ of regular constraints, and show a strong separation: Essentially, either every instance with $k$ tasks can be reduced to at most $k$ users, or there is no polynomial-time reduction to poly$(k)$ users unless the polynomial hierarchy collapses. (However, some technical issues arise for the infinite case.)

\paragraph{Our results.} Our main result is a dichotomy for the \WSP{($\Gamma$)} problem when $\Gamma$ contains only regular relations. We show two results. On the one hand, if every relation $R \in \Gamma$ is \emph{intersection-closed} (see Section \ref{section:dichotomy}), then we give a polynomial-time reduction which reduces the number of users in an instance to $n' \leq k$, without increasing the number of tasks $k$ or constraints $m$. This applies even if $\Gamma$ is infinite, given a natural assumption on computable properties of the relations. On the other hand, we show that given even a single relation $R$ which is regular but not intersection-closed, the problem \WSP{(\uppercase{$R$})} restricted to using only the relation $R$ admits no polynomial kernel, and hence no reduction to poly$(k)$ users, unless the polynomial hierarchy collapses.
For finite sets $\Gamma$, this gives a dichotomy in a straight-forward manner: For every finite set $\Gamma$ of regular relations, \WSP{($\Gamma$)} admits a polynomial kernel if every $R \in \Gamma$ is intersection-closed, and otherwise not unless the polynomial hierarchy collapses. 

However, for infinite sets~$\Gamma$ things get slightly more technical, for two reasons: (1) An instance with $k$ tasks and few users could still be exponentially large due to the number of constraints, analogously to the result that \probname{Hitting Set} admits no polynomial kernel parameterized by the size of the ground set \cite{DomLS09} (cf.~\cite{HermelinKSWW13}). (2) More degenerately, without any restriction on~$\Gamma$, an instance could be exponentially large simply due to the encoding size of a single constraint (e.g., one could interpret a complete \WSP{} instance on~$k$ tasks as a single constraint on these~$k$ tasks). Both these points represent circumstances that are unlikely to be relevant for practical \WSP{} instances. We make two restrictions to cope with this: (1) We allow the number~$m$ of constraints as an extra parameter, since it could be argued that~$m\leq \textrm{poly}(k)$ in practice. (2) We require that each constraint of arity~$r\leq k$ can be expressed by~poly$(r)$ bits. E.g., this allows unbounded arity forms of all standard constraints. Using this, we obtain a more general dichotomy: For any (possibly infinite) set $\Gamma$ of regular relations, \WSP{($\Gamma$)} admits a kernel of size poly($k+m$) if every $R \in \Gamma$ is intersection-closed, otherwise not, unless the polynomial hierarchy collapses.

Note that prior to our work there was no conjecture on how a polynomial kernel dichotomy for all regular constraints may look like (we cannot offer such a conjecture for the more general case of user-independent constraints).
The positive part follows by generalizing ideas of Crampton et al.~\cite{CrGuYeJournal}; 
the negative part is more challenging, and requires more involved arguments, especially to show the completeness of the dichotomy (i.e., that every relation $R$ which is regular but not intersection-closed can be used in our lower bounds proof; see Section~\ref{sec:regularlower}).

\paragraph{Organization.} We define \WSP{} formally and introduce a number of different constraint types, including regular constraints, in Section~\ref{section:preliminaries}. In Section~\ref{section:lowerbounds} we give several lower bounds for the kernelization of \WSP{($\Gamma$)}. In Section~\ref{section:dichotomy} we prove our main result, namely the dichotomy for regular constraints. We conclude in Section~\ref{section:conclusion}.

\section{Preliminaries}\label{section:preliminaries}

We define a \emph{workflow schema} to be a tuple~$(S,U,A,C)$, where~$S$ is the set of tasks in the workflow,~$U$ is the set of users,~$A\colon S\to 2^U$ assigns each task~$s\in S$ an \emph{authorization list}~$A(s)\subseteq U$, and~$C$ is a set of workflow constraints. For the instance $W^*$ of \WSP{} of the previous section, $S=\{s_1,s_2,s_3\},\ U=\{u_1,\ldots , u_6\},\ C=\{s_1=s_2,s_1\neq s_3,s_2\neq s_3\}$, and $A(s_1)=\{u_1,u_2,u_3\}$, $A(s_2)=\{u_1,u_4,u_5\}$, $A(s_3)=\{u_1,u_6\}$.
A \emph{workflow constraint} is a pair $c = (L,\Theta)$, where $L \subseteq S$ is the \emph{scope} of the constraint and $\Theta$ is a set of functions from $L$ to $U$ that specifies those assignments of elements of $U$ to elements of $L$ that \emph{satisfy} the constraint $c$. For the constraint $s_1=s_2$ above, $L=\{s_1,s_2\}$ and $\Theta=\{ L\rightarrow \{u_i\}:\ i\in \{1,\ldots, 6\}\}$.

Given $T \subseteq S$ and $X \subseteq U$, a \emph{plan} is a function $\pi\colon T \to X$; a plan $\pi\colon S \to U$ is called a \emph{complete plan}.
Given a workflow constraint $(L,\Theta)$, $T \subseteq S$, and $X \subseteq U$, a plan $\pi\colon T \to X$ \emph{satisfies} $(L,\Theta)$ if  either $L\setminus T\neq \emptyset$, or $\pi$ restricted to $L$ is contained in $\Theta$, i.e., defining $\pi_L \colon L \to U \colon s \mapsto \pi(s)$ we have $\pi_L \in \Theta$. 
A plan $\pi\colon T \to X$ is \emph{eligible} if $\pi$ satisfies every constraint in $C$. A plan $\pi\colon T \to X$ is \emph{authorized} if $\pi(s) \in A(s)$ for all $s \in T$.  A plan is \emph{valid} if it is complete, authorized and eligible. 
For an algorithm that runs on an instance $(S,U,A,C)$ of \WSP{}, we will measure the running time in terms of $n=|U|, k=|S|$, and $m=|C|$.

\subsection{WSP constraints and further notation} \label{sec:prels:wsp}

Let us first recall some concrete constraints that are of interest for this work:
\begin{description}
 \item[$(=,T,T'),(\neq,T,T')$:] These generalize the binary \emph{binding-of-duty} and \emph{separa\-tion-of-duty} constraints and were previously studied in~\cite{CrGuYeJournal,WangLi10}. They demand that there exist~$s\in T$ and~$s'\in T'$ which are assigned to the same (resp.\ different) users. We shorthand~$(s=s')$ and~$(s\neq s')$ if~$T=\{s\}$ and~$T'=\{s'\}$.
 \item[$(t_\ell,t_r,T)$:] A plan $\pi$ satisfies $(t_\ell,t_r,T)$, also called a \emph{tasks-per-user counting constraint}, if a user performs either no tasks in $T$ or between $t_\ell$ and $t_r$ tasks. Tasks-per-user counting constraints generalize the cardinality constraints which have been widely adopted by the \WSP{} community~\cite{ansi-rbac04,BeboFe01,JoBeLaGh05,SaCoFeYo96}.
 \item[$(\leq t,T),(\geq t,T)$:] These demand that the tasks in~$T$ are assigned to at most~$t$ (resp.\ at least~$t$) different users. They generalize \emph{binding-of-duty} and \emph{separa\-tion-of-duty}, respectively, and enforce security and diversity~\cite{CoCrGaGuJo14}.  
\end{description}
All these constraints share the property that satisfying them depends only on the partition of tasks that is induced by the plan. This property is referred to as \emph{user-independence}; see below.

\paragraph{Regular and user-independent constraints.} 
Formally, a constraint $(L,\Theta)$ is \emph{user-independent} if for any $\theta \in \Theta$  and any permutation $\psi \colon U \to U$, we have $\psi \circ \theta \in \Theta$. Note that this definition of user-independent constraints is equivalent to the (more informal) definition of these constraints given in the previous section.

For $T\subseteq S$ and $u\in U$ let $\pi\colon\ T\to u$ denote the plan that assigns every task of $T$ to $u$. A constraint $c=(L,\Theta)$ is \emph{regular} if it satisfies the following condition: For any partition $L_1,\ldots ,L_p$ of $L$ such that for every $i \in [p]=\{1,\ldots ,p\}$ there exists an eligible plan $\pi\colon L \rightarrow U$ and user $u$ such that $\pi^{-1}(u) = L_i$, the plan $\bigcup_{i =1}^p (L_i \to u_i)$, where all $u_i$'s are distinct, is eligible. Consider, as an example, a tasks-per-user counting constraint $(t_\ell,t_r,L)$. Let  $L_1,\ldots ,L_p$ be a partition of $L$ such that for every $i \in [p]$ there exists an eligible plan $\pi\colon L \rightarrow U$ and user $u$ such that $\pi^{-1}(u) = L_i$. Observe that for each $i\in [p],$ we have $t_{\ell} \leq |L_i|\leq t_r$ and so the plan $\bigcup_{i =1}^p (L_i \to u_i)$, where all $u_i$'s are distinct, is eligible. Thus, any tasks-per-user counting constraint $(t_\ell,t_r,L)$ is regular.
Regular constraints are a special class of user-independent constraints, but not every user-independent constraint is regular.
Crampton et al.\ \cite{CrGuYeJournal} show that constraints of the type $(\neq,T,T')$; $(=,T,T')$, where at least one of the sets $T,T'$ is a singleton are regular. In general, $(=, T, T')$ is not regular~\cite{CrGuYeJournal}: Consider, e.g.,~$(=,\{s_1,s_2\},\{s_3,s_4\})$, where we have eligible plans for every choice of assigning some $s_i$ to a private user, but assigning all four tasks to private users is ineligible.

Since regular constraints are of central importance to this paper, we introduce some further notation and terminology.
Below, we generally follow Crampton et al.~\cite{CrGuYeJournal}. 
Let $W=(S,U,A,C)$ be a workflow schema, and $\pi$ an eligible complete plan for $W$. 
Then $\sim_\pi$ is the equivalence relation on $S$ defined by $\pi$,
where $s \sim_\pi s'$ if and only if $\pi(s)=\pi(s')$. 
We let $S/\pi$ be the set of equivalence classes of $\sim_\pi$,
and for a task $s \in S$ we let $[s]_\pi$ denote the equivalence class containing $s$.

For a constraint $c=(L,\Theta)$, a set $T \subseteq L$ of tasks is \emph{$c$-eligible}
if there is a plan $\pi\colon L \to U$ that satisfies $c$, such that $T \in L/\pi$.
It is evident from the definition that $c$ is regular if and only if the following holds:
For every plan $\pi\colon L \to U$, $\pi$ satisfies $c$ if and only if every equivalence 
class $T \in L/\pi$ is $c$-eligible. 
In this sense, a regular constraint $c$ is entirely defined by the set of $c$-eligible sets of tasks.
It is clear that regular constraints are closed under conjunction, i.e., if $C$ is a set of regular constraints
on a set $T$ of tasks, then the set of plans $\pi \colon T \to U$ which are $c$-eligible for every $c \in C$
defines a new regular constraint $(T,\Theta)$. 

In a similar sense, if $c=(L,\Theta)$ is user-independent but not necessarily regular, 
then $c$ can be characterized on the level of partitions of $L$:
Let $\pi, \pi'\colon L \to U$ be two plans such that $L/\pi = L/\pi'$. 
Then either both $\pi$ and $\pi'$ are eligible for $c$, or neither is.
Overloading the above terminology, if $c$ is a user-independent constraint,
then we say that a partition $L/\pi$ is $c$-eligible if a plan $\pi$
generating the partition would satisfy the constraint.
We may thus refer to the partition $L/\pi$ itself as either eligible or ineligible. 
As with regular constraints, user-independent constraints are closed under conjunction. 

\paragraph{Describing constraints via relations.}
We will frequently describe constraint types in terms of relations. In the following, we restrict ourselves
to user-independent constraints. Let $R \subseteq \N^r$ be an $r$-ary relation, and $(s_1,\ldots,s_r) \in S^r$ a tuple of tasks,
with repetitions allowed (i.e., we may have $s_i=s_j$ for some $i \neq j$, $i,j \in [r]$). 
An \emph{application $R(s_1,\ldots,s_r)$ (of $R$)} is a constraint $(L, \Theta)$ where $L=\{s_i: i \in [r]\}$
and $\Theta=\{(\pi\colon L \to \N) : (\pi(s_1),\ldots,\pi(s_r)) \in R\}$. 
Here, we identify users $U=\{u_1,\ldots,u_n\}$ with integers $[n]=\{1,\ldots,n\}$.
We say that $R$ is \emph{user-independent (regular)} if every constraint $R(s_1,\ldots,s_n)$
resulting from an application of $R$ is user-independent (regular). 
In particular, a user-independent relation $R$ can be defined on the level of partitions,
in terms of whether each partition $L/\pi$ of its arguments is eligible or not,
and a regular relation can be defined in terms of eligible sets, as above.

Given a (possibly infinite) set $\Gamma$ of relations as above, a \emph{workflow schema over $\Gamma$}
is one where every constraint is an application of a relation $R \in \Gamma$,
and \WSP{($\Gamma$)} denotes the \WSP{} problem restricted to workflow schemata over $\Gamma$.
To cover cases of constraints of unbounded arity, we allow $\Gamma$ to be infinite.

For example, in the workflow schema $W^*$ of the previous section every constraint is an application of binary relations $=$ and $\neq$.

\paragraph{Well-behaved constraint sets.}
To avoid several degenerate cases associated with infinite sets $\Gamma$ we make some standard assumptions on our constraints. We say that a set $\Gamma$ of user-independent relations is \emph{well-behaved} if the following hold: (1) Every relation~$R\in\Gamma$ can be encoded using~$poly(r)$ bits, where~$r$ is the arity of~$R$; note that this does not include the space needed to specify the scope of an application of $R$. (2) For every application $c=(L,\Theta)$ of a relation $R \in \Gamma$, we can test in polynomial time whether a partition of $L$ is $c$-eligible; we can also test in polynomial time whether a set $S \subseteq L$ is $c$-eligible, and if not, then we can (if possible) find a $c$-eligible set $S'$ with $S \subset S' \subseteq L$. 
All relations corresponding to the concrete constraints mentioned above, are well-behaved.

\subsection{Parameterized complexity and kernelization}

A \emph{parameterized problem}~$\Q$ is a subset of~$\Sigma^*\times\N$ for some finite alphabet~$\Sigma$. A parameterized problem $\Q$ is \emph{fixed-parameter tractable} (\FPT) if there is a computable function $f\colon\N\to\N$ and an algorithm that, given $(x,k)$, takes time $f(k)|x|^{\Oh(1)}$ and correctly decides whether $(x,k)\in\Q$. A \emph{kernelization} of~$\Q$ is a polynomial-time computable function $K\colon (x,k) \mapsto (x',k')$ such that $(x,k) \in \Q$ if and only if $(x',k') \in \Q$, and such that $|x'|, k' \leq h(k)$ for some $h(k)$. Here, $(x,k)$ is an \emph{instance} of $\Q$, and $h(k)$ is the \emph{size} of the kernel. We say that $K$ is a \emph{polynomial kernelization} if~$h(k)=k^{\Oh(1)}$. For an introduction to parameterized complexity we refer to, e.g., \cite{DowneyF13}.

Our main tool for studying existence of polynomial kernels is kernelization-preserving reductions. Given two parameterized problems $\Q_1$ and $\Q_2$, a \emph{polynomial parametric transformation (PPT)} from $\Q_1$ to $\Q_2$ is a polynomial time computable function $\Psi \colon (x,k) \mapsto (x',k')$ such that for every input $(x,k)$ of $\Q_1$ we have $(x',k') \in \Q_2$ if and only if $(x,k) \in \Q_1$, and such that $k' \leq p(k)$ for some $p(k) = k^{\Oh(1)}$. Note that if $\Q_2$ has a polynomial kernel and if there is a PPT from $\Q_2$ to $\Q_1$, then $\Q_1$ has a \emph{polynomial compression}, i.e., a kernel-like reduction to an instance of a different problem with total output size $k^{\Oh(1)}$. Furthermore, for many natural problems (including all considered in this paper), we are able to complete these reductions using \NP-completeness to produce a polynomial kernel for $\Q_1$. Conversely, by giving PPTs from problems that are already known not to admit polynomial compressions (under some assumption) we rule out polynomial kernels for the target problems. For more background on kernelization we refer the reader to the recent survey by Lokshtanov et al.~\cite{LokshtanovMS12}.

\section{Lower bounds for kernelization}\label{section:lowerbounds}

In this section we begin our investigation of the preprocessing properties of the \probname{Workflow Satisfiability Problem}. We establish lower bounds against polynomial kernels for \WSP{} for several widely-used constraint types.
Like for many other problems, e.g., \probname{Hitting Set($n$)} or \probname{CNF SAT($n$)}, there is little hope to get polynomial kernels for \WSP{} when we allow an unbounded number of constraints of arbitrary arity, cf.~\cite{DellM14,DomLS09,HermelinKSWW13}. As an example, we give Lemma~\ref{lemma:cnfsatn}, whose proof uses a PPT from \probname{CNF SAT($n$)} to \WSP{($\geq 2$)} with only two users.

\begin{lemma} \label{lemma:cnfsatn}
  Let \WSP{($\geq 2$)} be the \WSP{} problem with constraints $(\geq 2, L)$ for task sets $L$ of arbitrary arity.
  Then \WSP{($\geq 2$)} admits no polynomial kernelization with respect to the number~$k$ of tasks unless the polynomial hierarchy collapses,
  even if the number of users is restricted to $n=2$.
\end{lemma}
\begin{proof}
  We give a PPT from \probname{SAT($n$)}, i.e., \probname{SAT} parameterized by the number $n$ of variables. The fact that this problem admits no polynomial kernelization or compression is due to work of Dell and van Melkebeek~\cite{DellM14}. Let an instance~$\phi$ of \probname{SAT($n$)} be given, and let~$n$ denote the number of variables in the CNF-formula~$\phi$. For ease of presentation, let the variables of~$\phi$ be~$x_1,\ldots,x_n$. (To recall, a CNF-formula is a conjunction of clauses, each of which is a disjunction of literals~$x_i$ or~$\neg x_i$.)
  
  We construct a \WSP{($\geq 2$)} instance with two users~$t$ and~$f$, which intuitively represent true and false assignment to literals of~$\phi$. For ease of reading, we state all constraints~$(\geq 2,L)$ by simply declaring the sets~$L$ to which they are applied. To begin, for each variable~$x_i$ we create two tasks~$s_i$ and~$\s_i$ and a constraint~$\{s_i,\s_i\}$; both users are authorized for all these tasks. Intuitively, assigning users to~$s_i$ and~$\s_i$ corresponds to setting~$x_i$ and~$\neg x_i$ to true or false; the constraint~$\{s_i,\s_i\}$ ensures that exactly one is true and one is false. Additionally, add one more task~$d$ for which only user~$f$ is authorized.  
  For each clause~$c$ of~$\phi$ we create the following set~$L_c$: Add to~$L_c$ the task~$s_i$ if~$c$ contains a literal~$x_i$, and the task~$\s_i$ if~$c$ contains a literal~$\neg x_i$. Additionally add the task~$d$.
  
  If~$\phi$ has a satisfying assignment then we get a valid plan~$\pi$ by~$\pi(d)=f$ and
  \begin{align*}
   \pi(s_i)=\begin{cases}
             t & \mbox{if $x_i$ is true,}\\
             f & \mbox{if $x_i$ is false,}
            \end{cases}&&
   \pi(\s_i)=\begin{cases}
             f & \mbox{if $x_i$ is true,}\\
             t & \mbox{if $x_i$ is false.}
            \end{cases}
  \end{align*}
  For each constraint~$L_c$, any satisfying assignment for~$c$ must assign true to some literal~$\ell_j\in\{x_i,\neg x_i\}$ for some~$x_i$. This corresponds directly to a task~$s_i$ or~$\s_i$ in~$L_c$ which is assigned user~$t$. Since~$d$ is always assigned user~$f$ this fulfills~$(\geq 2,L_c)$.
  
  Conversely, let~$\pi$ be a valid plan for the created \WSP{($\geq 2$)} instance. We create an assignment for~$\{x_1,\ldots,x_n\}$ by setting~$x_i$ to true if~$\pi(s_i)=t$ and to false otherwise. We already argued earlier that~$\pi(s_i)\neq\pi(\s_i)$ due to constraint~$\{s_i,\s_i\}$. Now let~$c$ be any clause of~$\phi$. In the corresponding set~$L_c$ we have the task~$d$ which must be assigned user~$f$. Due to the constraint~$(\geq 2, L_c)$ at least one other task in~$L_c$ must be assigned user~$t$. If this is a task~$s_i$ then~$c$ contains~$x_i$ and our defined assignment sets~$x_i$ to true, satisfying~$c$. If this is a task~$\s_i$ then~$c$ contains~$\neg x_i$. We know that~$\pi(\s_i)=t$, so~$\pi(s_i)=f$, implying that our assignment sets~$x_i$ to false, satisfying~$\neg x_i$ and~$c$.
  
  Thus, our reduction is correct.
  It is easy to see that the construction can be performed in polynomial-time and that the number of tasks is~$2n+1\leq poly(n)$. Thus, the PPT from \probname{SAT($n$)} to \WSP{($\geq 2$)} proves that the latter problem has no polynomial kernelization unless the polynomial hierarchy collapses.
\end{proof}

In our further considerations we will avoid cases like the above, by either taking~$m$ as an additional parameter or by restricting $\Gamma$ to be finite, which implies bounded arity (namely the maximum arity over the finitely many $R\in\Gamma$).
We also assume that all constraints are well-behaved (cf. Section~\ref{sec:prels:wsp}).
We then have the following, showing that bounding the number of users implies a polynomial kernel.

\begin{proposition} \label{prop:kernelsize}
  Let $\Gamma$ be a set of relations. If $\Gamma$ is finite, then \WSP{($\Gamma$)} has a polynomial kernel under parameter $(k+n)$;
  if $\Gamma$ is infinite but $\Gamma$ is well-behaved, then \WSP{($\Gamma$)} has a polynomial kernel under parameter $(k+m+n)$.
\end{proposition}
\begin{proof}
  An instance of \WSP{($\Gamma$)} is defined by describing its tasks, users, authorization lists and constraints. Note that the former three can be written down in space $\Oh(kn)$, hence it remains to describe the constraints. If $\Gamma$ is finite, then there is a maximum arity $r$ of any relation in $\Gamma$, and at most $|\Gamma| \cdot k^r=k^{\Oh(1)}$ possible constraints can be defined as applications of relations $R \in \Gamma$, and each constraint can be defined in short space (e.g., by an index into $\Gamma$ and $\Oh(r \log k)$ bits giving the scope of the constraint). Hence all constraints can be described in space polynomial in $k+n$. 

  If $\Gamma$ is infinite but well-behaved, then under the parameter $k+m+n$ it suffices to be able to describe each constraint in space poly$(n+k)$; this is possible by assumption. Hence in both cases, there is a simple encoding of an instance in space poly$(k+n)$ resp.\ poly$(k+n+m)$. 
\end{proof}

The following lemma addresses a special case of ternary constraint~$R(a,b,c)$ and proves that already \WSP{(\uppercase{$R$})} admits no polynomial kernelization in terms of~$k+m$. This lemma will be a cornerstone of the dichotomy in the following section. We also get immediate corollaries ruling out polynomial kernels in~$k+m$ 
for constraints~$(=,S,S')$ and~$(\leq t,S)$, since~$(=,\{a\},\{b,c\})$ and~$(\leq 2,\{a,b,c\})$ fulfill the requirement of the lemma.

The proof will be by a PPT from the problem \probname{Multi-colored Hitting Set($m$)} (\probname{MCHS}), which was considered in~\cite{DomLS09,HermelinKSWW13}.
The input is a vertex set $V$, a collection $\cH=\{E_1,\ldots,E_m\}$ of subsets of $V$, an integer $\ell$, and a function $\phi\colon V \rightarrow [\ell]$ which \emph{colors} each vertex of $V$ in one of $\ell$ colors. The task is to find a set $Q \subseteq V$ containing exactly one vertex of each color such that $Q \cap E_i \neq \emptyset$ for every $E_i \in \cH$. It follows from work of Dom et al.~\cite{DomLS09} that this problem admits no polynomial kernel or compression under the parameter $m$ unless the polynomial hierarchy collapses: A simple PPT from \probname{Hitting Set($m$)} works by making $\ell$ copies of each element and giving each copy a different color; since instances with $\ell>m$ are trivial for \probname{Hitting Set($m$)} we may restrict \probname{MCHS} to $\ell\leq m$ without harming this lower bound. Furthermore, the problem is complete for a kernelization hardness class known as $\mathsf{WK[1]}$, which is conjectured to imply further lower bounds~\cite{HermelinKSWW13}. 

We now proceed with the PPT. 

\begin{lemma} \label{lemma:ternary}
  Let $R(a,b,c)$ be a ternary user-independent constraint which is satisfied by plans with induced partition~$\{\{a,b\},\{c\}\}$ or~$\{\{a,c\},b\}$, but not by plans with partition~$\{\{a\},\{b\},\{c\}\}$.
  Then \WSP{(\uppercase{$R$})} does not admit a polynomial kernel with respect to parameter~$k+m$ unless the polynomial hierarchy collapses.
\end{lemma}

\begin{proof}
 We give a PPT from \probname{Multi-colored Hitting Set($m$)}, described above. 
Let an instance~$(V,\cH,\ell,\phi)$ of \probname{MCHS} be given, where~$\cH=\{E_1,\ldots,E_m\}$ with~$E_i\subseteq V$, $\ell\leq m$, and~$\phi\colon V\to [\ell]$. Let $V_j=\phi^{-1}(j)$, $j \in [\ell]$. 
We may assume that $m \geq 2$, or else solve the instance in polynomial time and return a corresponding dummy yes- or no-instance, and that~$V_j\neq \emptyset$ for each $j\in [\ell]$.
A solution to \probname{MCHS} is now a multi-colored hitting set $Q=\{v_1,\dots ,v_{\ell}\}$ where $v_i \in V_i$ for each $i \in [\ell]$ and $Q\cap E_i\neq \emptyset$ for every $i\in [m].$ 

  We start by letting the set of users be~$U:=V$. We make~$(\ell-1)\cdot m$ tasks
  \begin{align*}
  e_{1,2},\ldots,e_{1,\ell},\quad e_{2,2},\ldots,e_{2,\ell},\quad\ldots\quad,e_{m,2},\ldots,e_{m,\ell},
  \end{align*}
  i.e.,~$(\ell-1)$ tasks~$e_{i,2},\ldots,e_{i,\ell}$ for each set~$E_i$. For every $i\in [m]$, let the authorization of $e_{i,\ell}$ be $E_i$; all remaining tasks~$e_{i,j}$ get authorization~$U$. Furthermore, we introduce tasks~$s_1,\ldots,s_\ell$ that are intended for choosing a hitting set of size~$\ell$ and that have authorizations~$A(s_j):=V_j$.

  We introduce the following constraints for each~$i\in\{1,\ldots,m\}$:
  \begin{enumerate}
   \item Introduce~$R(e_{i,2},s_1,s_2)$.
   \item For all~$j\in\{3,\ldots,\ell\}$ introduce~$R(e_{i,j},e_{i,j-1},s_j)$.
  \end{enumerate}

  Clearly, this construction can be performed in polynomial time and the parameter values are number of tasks~$k=(\ell-1)\cdot m+\ell=\Oh(m^2)$ and number of constraints~$m'=(\ell-1)\cdot m=\Oh(m^2)$. Thus, to prove that this constitutes a PPT it remains to prove that the created \WSP{(\uppercase{$R$})} instance has a valid plan~$\pi$ if and only if~$(V,\cH)$ has a multi-colored hitting set.

  Let~$S\subseteq V$ be a multi-colored hitting set for~$(V,\cH)$. For each color~$j$ let~$v_j$ be the unique element in~$S\cap V_j$. We now create~$\pi$ and begin with~$\pi(s_j):=v_j$, consistent with the authorization of~$s_j$. We now define the user assignment for tasks associated with some set~$E_i\in\cH$. First of all, we note that~$S\cap E_i \neq \emptyset$ and thus we can arbitrarily select~$v_t\in S\cap E_i$. We now assign users as follows for~$j\in\{2,\ldots,\ell\}$:
  \begin{align*}
  \pi(e_{i,j})=\begin{cases}
                v_1 & \mbox{if $j<t$,}\\
                v_t & \mbox{if $j\geq t$.}
               \end{cases}
  \end{align*}
  We note that~$v_t\in E_i$ and thus~$v_t$ is authorized for task~$e_{i,\ell}$; all other tasks are authorized for all users anyway. It can be easily verified that~$R(e_{i,2},s_1,s_2)$ and~$R(e_{i,j},e_{i,j-1},s_j)$, for~$j\in\{3,\ldots,\ell\}$ hold for this assignment of users:
  \begin{itemize}
   \item $R(e_{i,2},s_1,s_2)$ is satisfied because~$\pi$ gives partition~$\{\{e_{i,2},s_2\},\{s_1\}\}$ if~$t=2$ and partition~$\{\{e_{i,2},s_1\},\{s_2\}\}$ otherwise.
  \end{itemize}
   It remains to check~$R(e_{i,j},e_{i,j-1},s_j)$ for~$j=\{3,\ldots,\ell\}$.
  \begin{itemize}
   \item If~$j\in\{3,\ldots,t-1\}$ then~$\pi(e_{i,j-1})=v_1$,~$\pi(e_{i,j})=v_1$, and~$\pi(s_j)=v_j$, satisfying~$R(e_{i,j},e_{i,j-1},s_j)$ with partition~$\{\{e_{i,j},e_{i,j-1}\},\{s_j\}\}$.
   \item If~$j=t$ then~$\pi(e_{i,j-1})=v_1$,~$\pi(e_{i,j})=v_t$, and~$\pi(s_j)=v_j=v_t$, satisfying~$R(e_{i,j},e_{i,j-1},s_j)$ with partition~$\{\{e_{i,j},s_j\},\{e_{i,j-1}\}\}$.
   \item If~$j\in\{t+1,\ldots,\ell\}$ then~$\pi(e_{i,j-1})=v_t$,~$\pi(e_{i,j})=v_t$, and~$\pi(s_j)=v_j\neq v_t$, satisfying~$R(e_{i,j},e_{i,j-1},s_j)$ with partition~$\{\{e_{i,j},e_{i,j-1}\},\{s_j\}\}$.
  \end{itemize}

  For the converse, assume that~$\pi$ is a valid plan for the created \WSP{(\uppercase{$R$})} instance. For~$j\in\{1,\ldots,\ell\}$ let~$v_j:=\pi(s_j)$ and note that~$v_j\in A(s_j)=V_j$. We claim that~$S:=\{v_1,\ldots,v_\ell\}$ is a multi-colored hitting set for~$(V,\cH)$ according to~$\phi$. Consider any set~$E_i\in\cH$ and recall that~$A(e_{i,\ell})=E_i$. We claim~$\pi(e_{i,\ell})\in S$, which would imply that~$S\cap E_i\neq\emptyset$. To prove this, we prove inductively that~$\pi(e_{i,j})\in \{v_1,\ldots,v_j\}$ for~$j\in\{2,\ldots,\ell\}$. (Recall that authorizations for all these other tasks associated with~$E_i$ are simply~$U=V$, i.e., we really need the property for~$j=\ell$ and~$e_{i,\ell}$.)
  \begin{enumerate}
   \item For~$j=2$ consider the constraint~$R(e_{i,2},s_1,s_2)$. Since~$A(s_1)\cap A(s_2)=V_1\cap V_2=\emptyset$ we know that~$\pi(s_1)\neq \pi(s_2)$. Furthermore, by assumption of the lemma~$R(e_{i,2},s_1,s_2)$ is not satisfied if~$\pi$ induces partition~$\{\{e_{i,2}\},\{s_1\},\{s_2\}\}$. Since~$\pi(s_1)\neq \pi(s_2)$ we must have~$\pi(e_{i,2})\in\{\pi(s_1),\pi(s_2)\}=\{v_1,v_2\}$. This proves our claim for~$j=2$.
   \item Now consider some~$j\geq 3$ such that the claim holds for all smaller~$j$. In particular~$\pi(e_{i,j-1})\in\{v_1,\ldots,v_{j-1}\}$. Note that~$\{v_1,\ldots,v_{j-1}\}\subseteq V_1\cup\ldots\cup V_{j-1}$ which has an empty intersection with~$V_j$. Thus, considering the constraint~$R(e_{i,j},e_{i,j-1},s_j)$ we find that~$\pi(e_{i,j-1})\neq \pi(s_j)$. Thus, by the same argument as in the previous item we find that~$\pi(e_{i,j})\in\{\pi(e_{i,j-1}),\pi(s_j)\}\subseteq\{v_1,\ldots,v_j\}$. This completes our claim.
  \end{enumerate}
  It follows that~$\pi(e_{i,\ell})\in S$, which implies that~$S\cap E_i=S\cap A(e_{i,\ell})\supseteq\{\pi(e_{i,\ell})\}\neq\emptyset$. Thus,~$S$ is indeed a multi-colored hitting set for~$(V,\cH)$ according to~$\phi$. Since \probname{Multi-colored Hitting Set($m$)} is known not to admit a polynomial compression unless the polynomial hierarchy collapses, this completes the proof.
\end{proof}

 Since both~$(=,\{a\},\{b,c\})$ and~$(\leq 2,\{a,b,c\})$ fulfill the requirement of the lemma, we get the following corollary.

\begin{corollary}\label{cor:lbs}
\WSP{($(=,S,S')$)} and \WSP{($(\leq t,S)$)} do not admit a kernelization to size polynomial in~$k+m$ unless the polynomial hierarchy collapses.
\end{corollary}

\section{A dichotomy for regular constraints}\label{section:dichotomy}

In this section, we present a dichotomy for the kernelization properties of \WSP{($\Gamma$)} when $\Gamma$ 
is a well-behaved set of regular relations.

Let us describe the dichotomy condition. Let $c=(L,\Theta)$ be a regular constraint, and $E_R \subseteq 2^L$ 
the set of $c$-eligible subsets of $L$; for ease of notation, we let $\emptyset \in E_R$. 
Note that by regularity, $E_R$ defines $c$. We say that $c$ is \emph{intersection-closed} if 
for any $T_1, T_2 \in E_R$ it holds that $T_1 \cap T_2 \in E_R$. Similarly, we say that a regular
relation $R \in \Gamma$ is intersection-closed if every application $R(s_1,\ldots,s_r)$ of $R$ is.
Note (1) that this holds if and only if an application $R(s_1,\ldots,s_r)$ of $R$ with $r$ distinct
tasks $s_i$ is intersection-closed, and (2) that the conjunction of intersection-closed constraints
again defines an intersection-closed constraint.
Finally, a set $\Gamma$ of relations is intersection-closed if every relation $R \in \Gamma$ is.
Our dichotomy results will essentially say that \WSP{($\Gamma$)} admits a polynomial kernel if and only if 
$\Gamma$ is intersection-closed; see Theorem~\ref{th:regulardich} below.

The rest of the section is laid out as follows. In Section~\ref{sec:regularupper}
we show that if $\Gamma$ is regular, intersection-closed, and well-behaved,
then \WSP{($\Gamma$)} admits a reduction to $n' \leq k$ users; by Prop.~\ref{prop:kernelsize},
this implies a polynomial kernel under parameter $(k+m)$, and under parameter $(k)$ 
if $\Gamma$ is finite. Section \ref{sec:lemma} is a short technical section, where we introduce a notion of user-independent relation 
{\em implemented} by a set of relations and prove Lemma \ref{lemma:implementppt}, both important for Section~\ref{sec:regularlower}.
In Section~\ref{sec:regularlower} we show that for any single
relation~$R$ that is not intersection-closed, the problem \WSP{(\uppercase{$R$})} admits no polynomial kernel, by application of 
Lemma~\ref{lemma:ternary}. In Section~\ref{sec:regularsummary} we consider 
the implications of these results for the existence of efficient user-reductions.

In summary, we will show the following result for kernelization. Again, a discussion of 
the consequences for user-reductions is deferred until Section~\ref{sec:regularsummary}.

\begin{theorem} \label{th:regulardich}
  Let $\Gamma$ be a possibly infinite set of well-behaved regular relations.
  If every relation in $\Gamma$ is intersection-closed, then \WSP{($\Gamma$)}
  admits a polynomial-time many-one reduction down to $n' \leq k$ users, implying a polynomial kernel
  under parameter $k+m$ (and a polynomial kernel under parameter $k$ if $\Gamma$ 
  is finite). Otherwise, \WSP{($\Gamma$)} admits no kernel of size poly$(k+m)$
  unless the polynomial hierarchy collapses (even if $\Gamma$ consists of a single
  such relation $R$).
\end{theorem}

\subsection{A user reduction for intersection-closed constraints}
\label{sec:regularupper}

We now give a procedure that reduces a \WSP{} instance $W=(S,U,A,C)$ with $n$ users,
$k$ tasks and $m$ constraints to one with $k' \leq k$ tasks, $n' \leq k'$ users
and $m' \leq m$ constraints, under the assumption that every constraint $c \in C$
occurring in the instance is intersection-closed and that our language is well-behaved.
(This has been called a \emph{partial kernel} in other work \cite{BetzlerBN10}.)
The approach is as in Crampton et al.~\cite[Theorem~6.5]{CrGuYeJournal},
but becomes more involved due to having to work in full generality;
we also use a more refined marking step 
that allows us to decrease the number of users from $k^2$ to $k$, a significant improvement.
As noted (Prop.~\ref{prop:kernelsize}), under the appropriate further assumption 
on the constraints, this gives a polynomial kernel under parameter $k+m$ or $k$.

We begin by noting a consequence of sets closed under intersection. 

  \begin{lemma} \label{lemma:hornprop}
    Let $c=(L,\Theta)$ be an intersection-closed constraint,
    and let $T \subseteq L$ be $c$-ineligible.
    If there is a superset $T'$ of $T$ which is $c$-eligible,
    then there is a task $s \in L \setminus T$ 
    such that every $c$-eligible superset $T'$ of $T$ contains $s$. 
  \end{lemma}
  \begin{proof}
    Let $T_{\cap}$ be the intersection of all $c$-eligible supersets $T' \supset T$.
    Then $T_{\cap}$ must itself be $c$-eligible. 
    We clearly have $T \subseteq T_{\cap}$, and since $T$ is $c$-ineligible the containment
    must be strict. Hence there is some $s \in T_{\cap} \setminus T$;
    this task $s$ must be contained in all $c$-eligible supersets of $T$. 
  \end{proof}

We refer to the task $s$ guaranteed by the lemma as a \emph{required addition} to $T$ by~$c$. 
Note that assuming well-behavedness, we can make this lemma constructive, 
i.e., in polynomial time we can test whether a set $T$ is eligible for a constraint, 
whether it has an eligible superset, and find all required additions if it does.
This can be done by first asking for an eligible superset $T'$ of $T$, then greedily
finding a minimal set $T \subset T'' \subseteq T'$. Then every $s \in T'' \setminus T$ is a required addition.

Our reduction proceeds in three phases. First, we detect all binary equalities implied by the constraints
i.e., all explicit or implicit constraints $(s=s')$, and handle them separately by 
merging tasks, intersecting their authorization lists. The output of this phase is an instance
where any plan which assigns to every task a unique user is eligible (though such a plan may 
not be authorized); in particular, since our constraints are regular, we have that
all singleton sets of tasks are $c$-eligible for every constraint $c$ of the instance. 

The second phase of the kernel is a user-marking process, similar to the kernels in~\cite{CrGuYeJournal}
but with a stronger bound on the number of users. This procedure is based around attempting 
to produce a system of \emph{distinct representatives} for $\{A(s): s \in S\}$,
i.e., to find a plan $\pi\colon S \to U$ such that $\pi$ is authorized and $\pi(s) \neq \pi(s')$ for every $s \neq s'$.
Via Hall's theorem, this procedure either succeeds, or produces a set $T$ of tasks such
that fewer than $|T|$ users are authorized to perform any task in $T$. 
In the latter case, we mark all these users, discard the tasks $T$, and repeat the procedure.
Eventually, we end up with a (possibly empty) set of tasks $S'$ which allows for
a set of distinct representatives, and mark these representatives as well. 
Refer to a task $s$ as \emph{easy} if it was appointed a representative in this procedure,
and \emph{hard} if it was not (i.e., if it was a member of a set $T$ of discarded tasks). 
We discard every non-marked user, resulting in a partially polynomial kernel with $k' \leq k$ tasks 
and $n' \leq k' \leq k$ users.

Finally, to establish the correctness of the kernelization, we give a procedure that, 
given a partial plan for the set of hard tasks, either extends the plan to a valid 
plan 
or derives that no such extended plan exists. 

\begin{lemma} \label{lemma:merge}
  Let $(S,U,A,C)$ be a workflow schema with $k$ tasks, $n$ users, and $m$ constraints,
  with at least one equality constraint $(s=s')$, $s \neq s'$. 
  In polynomial time, we can produce an equivalent instance $(S',U',A',C')$
  with at most $k-1$ tasks, $n$ users, and $m$ constraints.
  Furthermore, if the constraints in $C$ were given as applications $R(\ldots)$
  of some relations $R$, $R \in \Gamma$, then the constraints in $C'$ 
  can be given the same way.
\end{lemma}
\begin{proof}
  Drop the constraint $(s=s')$. 
  For every other constraint $c=(L,\Theta)$ with $s'$ in the scope, replace $c$
  by the corresponding constraint produced by replacing $s'$ by $s$.
  (If $c=R(s_1,\ldots,s_r)$ for some relation $R$, then this produces a new application 
  $R(\ldots)$ of the same relation $R$. This application may contain the task $s$ 
  in more than one position, however, this is allowed by our model of constraints.)
  Update the authorization list so that $A(s) := A(s) \cap A(s')$. 
  Finally, discard the task $s'$. The new instance has a valid plan 
  if and only if the old instance does. 
\end{proof}

We now show the detection of equalities.

\begin{lemma} \label{lemma:noeq}
  Let $(S,U,A,C)$ be a workflow schema where every constraint is regular and intersection-closed.
  Then we can in polynomial time reduce the instance to the case where every singleton $\{s\}$,
  $s \in S$, is eligible. 
\end{lemma}
\begin{proof}
We provide a procedure that, using calls to Lemma~\ref{lemma:hornprop}, detects all equality
constraints $(s=s')$ implied by the schema, and applies Lemma~\ref{lemma:merge} for
every such constraint found. 
The procedure goes as follows:
\begin{enumerate}
\item For every task $s \in S$, check whether the singleton set $\{s\}$ is $c$-eligible
  for every constraint $c$ in the instance.
\item If all such singleton task sets are eligible, then we are done.
\item Otherwise, let $\{s\}$ be ineligible for some constraint $c=(L,\Theta)$, 
  and let $s'$ be a required addition to $\{s\}$ by $c$. 
  Apply Lemma~\ref{lemma:merge} to the constraint $(s=s')$ and restart from Step 1.
\end{enumerate}

We now show the correctness of the procedure.
On the one hand, it is clear that by the termination of the procedure, no equality constraints
can remain, implicit or explicit, since a constraint $(s=s')$ contradicts that the singletons
$\{s\}$ and $\{s'\}$ are eligible (or, more formally, that the partition $\{\{s\}: s \in S\}$
is eligible for the remaining tasks $S'$). On the other hand, if $s'$
is a required addition to $\{s\}$ by some constraint $c$, then $c$
must imply the constraint $(s=s')$, since in this case, 
by Lemma~\ref{lemma:hornprop}, every $c$-eligible set containing $s$ also contains~$s'$. 
\end{proof}

Next, we describe the user-marking procedure in detail. We assume that Lemma~\ref{lemma:noeq}
has been applied, i.e., that all singleton sets are eligible. 
\begin{enumerate}
\item Let $M=\emptyset$, let $S$ be the set of all tasks, and $U$ the set of all users.
\item \label{step:loop} While $\{A(s) \cap U: s \in S\}$ does not admit a system of distinct representatives:
Let $T \subseteq S$ such that $|\bigcup_{s \in T} A(s)|<|T|$. Let $U_T=\bigcup_{s \in T} A(s)$. 
Add $U_T$ to $M$, remove $U_T$ from $U$, and remove from $S$ every task $s$ such that $A(s) \subseteq M$. 
\item \label{step:private} Add to $M$ the distinct representatives of the remaining tasks $S$, if any.
\item Discard all users not occurring in $M$ from the instance.
\end{enumerate}
We refer to the set $M$ of users produced above as the \emph{marked} users,
and let $S_{\textrm{hard}} \subseteq [k]$ be the set of hard tasks, i.e.,
the set of tasks removed in Step~\ref{step:loop} of the procedure.
Finally, we show the correctness of the above procedure.

\begin{lemma} \label{lemma:correctness}
  Let $(S,U,A,C)$ be a workflow schema where all constraints are regular and intersection-closed,
  and where all singleton sets are eligible. There is a valid plan for the instance
  if and only if there is a valid plan only using marked users.
\end{lemma}
\begin{proof}
  We describe a procedure that, for any eligible and authorized partial plan assigning users to the hard tasks,
  either produces a valid plan using only marked users or proves that the plan
  cannot be extended to a complete plan. Note that necessarily, the partial plan
  we begin with can use only marked users.
  
  The algorithm works as follows. Let $\pi$ be the partial plan.
\begin{enumerate}
\item Let $\cP= S_{\textrm{hard}}/\pi$ be the partition induced by the partial plan $\pi$.
\item Repeat the following until all sets $T \in \cP$ are $c$-eligible for every $c \in C$:
  \begin{enumerate}
  \item \label{step:a} Let $T$ be an ineligible set in $\cP$, and let $s$ be a required addition to $T$.
  \item \label{step:b} If $s$ is already assigned by $\pi$, or if $A(s) \not \ni\pi(s')$, $s' \in T$, then reject.
  \item \label{step:c} Otherwise, add $s$ to $T$ in $\cP$ and update $\pi$ (i.e., $\pi(s)=\pi(s')$, $s' \in T$).
  \end{enumerate}
\item \label{step:d} Pad the partition $\cP$ with singleton task sets, i.e., let every unassigned task
  be performed by a user which has no other duties. Update $\pi$ accordingly to a complete plan
  using the distinct representatives for the remaining tasks.
\end{enumerate}
We show correctness of this procedure. First, by Lemma~\ref{lemma:hornprop}, clearly the 
modification in Steps \ref{step:a}--\ref{step:c} are necessary for any valid plan. 
It also follows that any rejection performed during these steps will be correct.
Second, the padding in Step~\ref{step:d} is possible, since the users appointed as distinct 
representatives will be used for no other task.
Hence, we can create a complete plan $\pi$ corresponding to the resulting padded partition $\cP$,
and this plan will be valid, as every task set $T \in \cP$ is either
a singleton set, and hence $c$-eligible for every $c$, or the output of the loop \ref{step:a}--\ref{step:c}
and thus $c$-eligible for every $c$ by assumption; furthermore, authorization was checked at every step.
\end{proof}

Putting the above pieces together yields the following theorem.

\begin{theorem} \label{th:kernel}
  The \probname{Workflow Satisfiability Problem}, restricted to well-behaved constraint languages
  where every constraint is  regular and intersection-closed admits a  kernel
  with $m' \leq m$ constraints, $k' \leq k$ tasks, and $n' \leq k'$ users. 
\end{theorem}
\begin{proof}
  The polynomial running time and the bound $k' \leq k$ are immediate from the above,
  and the correctness has already been argued. It remains to show that the number
  of marked users is at most equal to the number of tasks.
  This follows inductively, since the final addition to $M$, in Step~\ref{step:private}
  of the marking procedure, adds exactly $|S|$ users, and every previous addition,
  in an iteration of Step~\ref{step:loop}, adds fewer users to $M$ than the number of tasks in $T$,
  all of which will be removed from the set $S$. 
  Finally, the result is a kernel by Prop.~\ref{prop:kernelsize}.
\end{proof}

\subsection{Implementations and implications}\label{sec:lemma}

Let $W=(S,U,A,C)$ be a workflow schema and $T \subseteq S$ a set of tasks.
The \emph{projection of $W$ onto $T$} is a constraint~$c=(T,\Theta)$ where~$\theta \colon T \to U$
is contained in~$\Theta$ if and only if there is a valid plan~$\pi$ for $W$ that extends $\theta$. Consider the \WSP{} instance $W^*$ introduced in Section \ref{sec:intro} and
let $T=\{s_2,s_3\}$. Since the only valid plan is $\pi$ with $\pi(s_1)=\pi(s_2)=u_1$ and $\pi(s_3)=u_6$, the projection
 of $W$ onto $T$ is $(T,\{\theta\})$, where $\theta$ is defined by $\theta(s_2)=u_1, \theta(s_3)=u_6.$

Further, let $R \subseteq \N^r$ be a user-independent relation, $Q=\{q_1,\ldots,q_r\}$ a set of $r$
distinct tasks, and $\Gamma$ a set of relations. 
We say that $\Gamma$ \emph{implements} $R$ if, for any $r$-tuple $\cA=(A(q_1), \ldots, A(q_r))$
of authorization lists, there is a workflow schema $W=(S,U,A,C)$ over $\Gamma$ that can be computed in polynomial time,
such that the projection of $W$ onto $T$ for some $T \subseteq S$ is equivalent to $R(q_1,\ldots,q_r)$ 
for every plan $\pi\colon \{q_1,\ldots,q_r\} \to U$ authorized with respect to $\cA$,
where furthermore~$|S|+|C|$ does not depend on $\cA$ and $U$ equals $\bigcup_{i \in [r]} A(q_i)$ 
plus a constant number of local users, i.e., new users who will not be authorized to perform any task outside of 
$S \setminus T$.

\begin{lemma} \label{lemma:implementppt}
  Let $\Gamma$ and $\Gamma'$ be finite workflow constraint languages
  such that $\Gamma'$ implements $R$ for every $R \in \Gamma$. 
  Then there is a PPT from \WSP{($\Gamma$)} to \WSP{($\Gamma'$)}, both with respect to
  parameter $k$ and $k+m$.
\end{lemma}
\begin{proof}
  Let $W=(S,U,A,C)$ be an instance of \WSP{($\Gamma$)}. We will create an equivalent instance $W'$ of \WSP{($\Gamma'$)}, which we will refer to as the \emph{output} of the reduction. For every constraint $c \in C$ which is an application of a relation $R \in \Gamma \setminus \Gamma'$, let $W_c$ be an implementation of $c$ (using the relevant authorization lists $A(s)$ from $A$). Add this implementation to the output, ensuring that all tasks and users local to $W_c$ are chosen distinct from existing tasks and users. Also add to the output every constraint $c \in C$ which is an application of a relation $R \in \Gamma \cap \Gamma'$. Clearly, this creates an output instance $W'$ of \WSP{($\Gamma'$)} which is equivalent to $W$, and which is computed in polynomial time. Furthermore, if $W$ contains $k$ tasks, $n$ users and $m$ constraints, then $W'$ contains $\Oh(k+m)$ tasks, $\Oh(m)$ constraints, and $\Oh(n+m)$ users. Hence it is a PPT under parameter $k+m$. 
Finally, since $\Gamma$ is finite, we have $m=k^{\Oh(1)}$ and the reduction is a PPT under parameter $k$. 
\end{proof}

\subsection{Kernel lower bounds for non-intersection-closed constraints} \label{sec:regularlower}

We now give the other side of the dichotomy by showing that within the setting of regular constraints,
even a single relation $R$ which is not intersection-closed can be used to construct
a kernelization lower bound, following one of the constructions in Section~\ref{section:lowerbounds}.
First, we need an auxiliary lemma.

\begin{lemma} \label{lemma:notboth}
  Let $R$ be a (satisfiable) regular relation $R$ which is not intersection-closed, 
  let $=$ be the binary equality relation, and let $\Gamma=\{R,=\}$. 
  Then either $\Gamma$ implements a relation matching that of Lemma~\ref{lemma:ternary},
  or $\Gamma$ implements the binary disequality relation $\neq$.
\end{lemma}
\begin{proof}
  Assume that $R$ is $r$-ary, and let $L=\{s_1,\ldots,s_r\}$ be a set of $r$ distinct tasks.
  We consider an application $c=R(s_1,\ldots,s_r)$ of $R$. Let $E_R$ be the $c$-eligible subsets of $L$. 
  First, if possible, let $T,T' \in E_R$ be disjoint sets such that there is a $c$-eligible plan $\pi\colon L \to U$
  with $T, T' \in L/\pi$, but $T \cup T' \notin E_R$. 
  Merge the tasks $T$ into a single task $s$ and the tasks $T'$ into a single task $t$
  by applications of $=$ (e.g., add constraints $(s_i=s_j)$ for every pair $s_i, s_j \in T$ and every pair $s_i, s_j \in T'$).
  Let $A(s)$ and $A(t)$ be arbitrary but non-empty, and let $A(s')=U_\bot$ for any other task $s'$, where $U_\bot$ is a 
  sufficiently large supply of dummy users (i.e., a set of users who are only authorized to perform 
  these tasks $s' \in L$). Then the resulting workflow schema $W$ has a valid plan $\pi$
  where $\pi(s)=u_s$ and $\pi(t)=u_t$ for any $u_s \in A(s)$ and $u_t \in A(t)$ with $u_s \neq u_t$,
  but no such plan where $u_s=u_t$. Hence we are done. 
  
  Second, if not, then we have that in any $c$-eligible partition $L/\pi$,
  we may freely merge parts to create a coarser $c$-eligible partition.
  In particular, for any $T \in E_R$ we have $(L \setminus T) \in E_R$.
  If there is any pair of sets $T, T' \in E_R$
  such that $T \cup T' \notin E_R$, then we proceed as follows.
  Let $P = L \setminus T$ and $Q = L \setminus T'$; then $P, Q \in E_R$
  but $P \cap Q \notin E_R$. Merge each of the sets $P \cap Q$,
  $P \setminus Q$, $Q \setminus P$, and $L \setminus (P \cup Q)$ 
  into single tasks, respectively, $a$, $b$, $c$ and $d$. 
  Let $R(a,b,c,d)$ be the resulting relation.
  Then $R$ defines a regular relation where $\{a\}$ is ineligible
  and where partitions $\{\{a,b\},\{c,d\}\}$ and $\{\{a,c\},\{b,d\}\}$
  are both eligible (by the eligible partitions $\{P, L \setminus P\}$ and $\{Q, L \setminus Q\}$).
  We will implement a relation matching Lemma~\ref{lemma:ternary}.
  Consider first the ternary relation $R'(a,b,c)= \exists d: R(a,b,c,d)$. 
  If $\{\{a\},\{b\},\{c\}\}$ is an ineligible partition for $R'$, then $R'$ is the relation we seek.
  Otherwise, we must have $\{a,d\} \in E_R$, hence also $\{b\}, \{c\} \in E_R$.   
  If $\{d\} \in E_R$, then we may restrict $A(d)$ to a supply of dummy users and again
  implement $R'(a,b,c)$ as in Lemma~\ref{lemma:ternary} on the remaining tasks.
  In the remaining case, using that $E_R$ is closed under complementation, 
  we find that $R$ is characterized precisely by excluding $\{a\}$ and $\{d\}$ 
  as eligible sets. In this case, $R'(a,b,c) = \exists d: R(b,a,d,c)$ produces
  a ternary relation $R'$ meeting the conditions of Lemma~\ref{lemma:ternary}.

  Otherwise, finally, $E_R$ is closed under arbitrary union and under complementation,
  and hence also under intersection, as  $T \cap T' = \overline{\overline{T} \cup \overline{T'}}$. 
  This contradicts our assumptions about $R$. 
\end{proof}

Using $\neq$, we can more easily construct a relation $R(a,b,c)$ as in Lemma~\ref{lemma:ternary}.

\begin{lemma} \label{lemma:withnotboth}
  Let $R$ be a regular relation which is not intersection-closed, and let $=$ and $\neq$ denote
  the binary equality and disequality relations.
  Then $\Gamma=\{R, =, \neq\}$ implements a relation $R'(a,b,c)$ as in Lemma~\ref{lemma:ternary}.
\end{lemma}
\begin{proof}
Let $c=R(L)$ be an application of $R$ on task set $L$, and let $E_R \subseteq 2^L$ be the set of $c$-eligible subsets of $L$. 
Say that $R$ is a \emph{counterexample relation} with respect to tasks $a, b, c \in L$ 
if there is a set $P \in E_R$ with $P \cap \{a,b,c\}=\{a,b\}$ and a set $Q \in E_R$ with $Q \cap \{a,b,c\}=\{a,c\}$,
but $P \cap Q \notin E_R$. We will modify $R$ into a simpler counterexample relation, 
by imposing further regular constraints using binary relations $=$ and $\neq$,
to finally produce the sought-after relation $R'(a,b,c)$. 

The first modification we consider is to merge tasks $s, s' \in L$.
If there is any pair of tasks $s, s'$ such that merging $s$ and $s'$ (i.e., adding an equality
constraint $(s=s')$) still yields a counterexample relation, then merge $s$ and $s'$; repeat this exhaustively.
For the rest of the proof, we will treat merged tasks as a single task, and assume that $R$ is a minimal counterexample
with respect to merging operations. In particular, this implies $|P \cap Q|=1$ for any sets $P$ and $Q$ as above.

Next, we impose a set of binary disequality constraints $(s \neq s')$.
Let $P$ and $Q$ be sets as above, let $P \cap Q = \{a\}$, and let $P' \subseteq P$
be a minimal eligible set with $a \in P'$. Similarly let $Q' \subseteq Q$ 
be a minimal eligible set with $a \in Q'$. Impose a binary disequality constraint 
$(a \neq s)$ for any $s \in L \setminus (P' \cup Q')$,
and $(b' \neq c')$ for any $b' \in P'-a$ and $c' \in Q'-a$.
We argue that the resulting workflow, projected down to tasks $a,b'',c''$ for $b'' \in P'-a$ 
and $c'' \in Q'-a$, satisfies the conditions of Lemma~\ref{lemma:ternary}.
It is clear by $P'$ and $Q'$ that the partitions $\{\{a, b''\},\{c''\}\}$ and $\{\{a, c''\},\{b''\}\}$ 
are eligible; the added disequality constraints have no effect on this. 
On the other hand, there is no eligible partition where $a, b'', c''$ are all contained in 
different sets, as due to the disequality constraints, $a$ can only be contained in 
a set contained in either $P'$ or in $Q'$, and both $P'$ and $Q'$ were chosen as minimal.
Hence we have implemented a relation $R(a,b'',c'')$ as in Lemma~\ref{lemma:ternary}.
\end{proof}

\begin{theorem} \label{th:dichlower}
  Let $R$ be a regular relation which is not intersection-closed.
  Then \WSP{(\uppercase{$R$})} admits no kernel of size poly$(k+m)$
  unless the polynomial hierarchy collapses. 
\end{theorem}
\begin{proof}
  By Lemma~\ref{lemma:implementppt}, it suffices to show that $\Gamma=\{R\}$
  implements a relation $R'$ matching the conditions of Lemma~\ref{lemma:ternary}.
  First consider the set $\Gamma'=\{R,=\}$.
  By Lemma~\ref{lemma:notboth}, $\Gamma'$ implements either
  a relation $R'$ matching the conditions of Lemma~\ref{lemma:ternary}
  or the binary disequality relation $\neq$; in the latter case,
  we have that $\Gamma''=\{R,=,\neq\}$ implements a relation $R'$
  by Lemma~\ref{lemma:withnotboth}, hence so does $\Gamma'$
  by the transitivity of implementations. 
  We find that \WSP{(\uppercase{$R,=$})} admits no polynomial kernel unless the polynomial hierarchy collapses.

  Finally, there is a trivial PPT from \WSP{(\uppercase{$R,=$})} to \WSP{(\uppercase{$R$})} by merging tasks:
  Let $(s=s')$ be an equality constraint in an instance of \WSP{(\uppercase{$R,=$})}.
  We may then apply Lemma~\ref{lemma:merge} to $s$ and $s'$, producing an equivalent
  instance with modified authorization lists and fewer tasks. 
  Repeating this until no equality constraints remain in the instance 
  yields an instance of \WSP{(\uppercase{$R$})}. 
  Hence the same lower bound applies to \WSP{(\uppercase{$R$})}.
\end{proof}

This finishes the proof of Theorem~\ref{th:regulardich}. 

\begin{corollary} \label{cor:dich}
  Let $\Gamma$ be a set of regular relations. If $\Gamma$ is well-behaved, then \WSP{($\Gamma$)} admits a polynomial kernel in parameter $k+m$ if $\Gamma$ is intersection-closed, otherwise not, unless the polynomial hierarchy collapses. If $\Gamma$ is finite, then the same dichotomy holds for parameter $k$ instead of $k+m$.
\end{corollary}

\subsection{On user bounds for WSP}\label{section:userbound} \label{sec:regularsummary}

In this section we return to the question of preprocessing \WSP{} down to a number of users that is polynomial in the number~$k$ of tasks. As seen above, the positive side of our kernel dichotomy relies directly on a procedure that reduces the number of users in an instance, while the lower bounds refer entirely to the \emph{total size} of the instance. Could there be a loophole here, allowing the number of users to be bounded without directly resulting in a polynomial kernel? Alas, it seems that while such a result cannot be excluded, it might not be very useful.

\begin{corollary}\label{cor:nouserred}
  Let $\Gamma$ be a set of user-independent relations containing at least one relation which is regular but not intersection-closed. 
  Unless the polynomial hierarchy collapses, any polynomial-time procedure that reduces the number of users in an instance down to poly$(k)$ must in some cases increase either the number $k$ of tasks, the number $m$ of constraints, or the coding length of individual constraints superpolynomially in $k+m$. 
\end{corollary}
\begin{proof}
  Let $R \in \Gamma$ be a constraint that is not intersection-closed. By Theorem~\ref{th:dichlower}, under our assumption there is no polynomial-time procedure that reduces every instance of \WSP{(\uppercase{$R$})} down to size $p(k+m)$ for any $p(t)=t^{\Oh(1)}$. Also note that a natural encoding of \WSP{} instances has coding length $\Oh(kn+km\ell )$, where $\ell$ is the largest coding length of an individual constraint. Hence, a procedure which bounds $n=k^{\Oh(1)}$ must sometimes increase one of the parameters $k$, $m$ and $\ell$ to $(k+m)^{\omega(1)}$. 
\end{proof}

\section{Conclusion} \label{section:conclusion}

\setlength{\tabcolsep}{0.12cm}

\begin{table}[t]
\centering
\begin{tabular}{@{}lcccccccc@{}}
\toprule
 & regular & 
$\cap$-closed &
\multicolumn{2}{c}{\begin{tabular}{@{}c@{}}poly($k$) user\\reduction\end{tabular}} &
\multicolumn{2}{c}{\begin{tabular}{@{}c@{}}bounded arity\\ resp.\ finite $\Gamma$ \end{tabular}} & 
\multicolumn{2}{c@{}}{\begin{tabular}{@{}c@{}}well-behaved\\infinite $\Gamma$\end{tabular}}\\
\cmidrule[\heavyrulewidth](lr){2-2}
\cmidrule[\heavyrulewidth](lr){3-3}
\cmidrule[\heavyrulewidth](lr){4-5}
\cmidrule[\heavyrulewidth](lr){6-7}
\cmidrule[\heavyrulewidth](l){8-9}
\begin{tabular}{@{}l@{}}$(\neq,T,T')$\\ $ (\geq 2,T)$ \end{tabular} & 
yes & 
yes & 
yes &
\cite{CrGuYeJournal}& 
PK($k$) &
\cite{CrGuYeJournal}&  
PK($k$+$m$) &
Cor.\ref{cor:dich}\\
\cmidrule(lr){2-2}
\cmidrule(lr){3-3}
\cmidrule(lr){4-5}
\cmidrule(lr){6-7}
\cmidrule(l){8-9}
\begin{tabular}{@{}l@{}}$(1,t_u,T)$\\ $(t_l,t_u,T)$ \end{tabular} & 
yes & 
\begin{tabular}{@{}c@{}}yes \\ no\end{tabular} & 
\begin{tabular}{@{}c@{}}yes\\ no \end{tabular} &
\begin{tabular}{@{}c@{}}\cite{CrGuYeJournal}\\Cor.\ref{cor:nouserred}\end{tabular} &  
\begin{tabular}{@{}c@{}}PK($k$)\\ no PK($k$+$m$) \end{tabular} &
\begin{tabular}{@{}c@{}}\cite{CrGuYeJournal}\\Thm.\ref{th:regulardich}\end{tabular} & 
\begin{tabular}{@{}c@{}} PK($k$+$m$) \\ no PK($k$+$m$) \end{tabular} & 
Cor.\ref{cor:dich}\\
\cmidrule(lr){2-2}
\cmidrule(lr){3-3}
\cmidrule(lr){4-5}
\cmidrule(lr){6-7}
\cmidrule(l){8-9}
\begin{tabular}{@{}l@{}}$(=,s,T')$\\ $(=,T,T')$ \end{tabular} & 
\begin{tabular}{@{}c@{}}yes\\ no \end{tabular} & 
\begin{tabular}{@{}c@{}}no\\n.a.\ \end{tabular} & 
no & 
Cor.\ref{cor:nouserred}& 
no PK($k$+$m$) &
Thm.\ref{th:kernel}& 
no PK($k$+$m$) &
Cor.\ref{cor:dich}\\
\cmidrule(lr){2-2}
\cmidrule(lr){3-3}
\cmidrule(lr){4-5}
\cmidrule(lr){6-7}
\cmidrule(l){8-9}
$(\geq t,T)$  & no & n.a.\ & yes & \cite{CrGuYeJournal} & PK($k$) & \cite{CrGuYeJournal} & PK($k$+$m$)&Prop.\ref{prop:kernelsize}\\
%
% \midrule
%
$(\leq t,T)$ & no & n.a.\ & no & Cor.\ref{cor:nouserred} & no PK($k$+$m$) & Cor.\ref{cor:lbs} & no PK($k$+$m$)& Cor.\ref{cor:lbs}\\
\cmidrule(lr){2-2}
\cmidrule(lr){3-3}
\cmidrule(lr){4-5}
\cmidrule(lr){6-7}
\cmidrule(l){8-9}
reg.+$\cap$-cl. & yes & yes & yes &Thm.\ref{th:kernel}& PK($k$) & Thm.\ref{th:regulardich} & PK($k+m)$ & Cor.\ref{cor:dich}\\
regular & yes & no & no & Cor.\ref{cor:nouserred} & no PK($k$+$m$) & Thm.\ref{th:regulardich} & no PK($k$+$m$) & Cor.\ref{cor:dich}\\
\bottomrule
\end{tabular}
\caption{\label{table:overview}Overview of results for typical user-independent constraints. We recall that the \WSP{} problem is \FPT with respect to~$k$ when all constraints are user-independent.}%
\end{table}

In this paper, we have considered kernelization properties of the workflow satisfiability problem \WSP{($\Gamma$)} restricted to use only certain types $R \in \Gamma$ of constraints. We have focused on the case that all relations $R \in \Gamma$ are regular. For this case, we showed that \WSP{($\Gamma$)} admits a reduction down to $n' \leq k$ users if every $R \in \Gamma$ is \emph{intersection-closed} (and obeys some natural assumptions on efficiently computable properties), otherwise (under natural restrictions) no such reduction is possible unless the polynomial hierarchy collapses. In particular, this implies a dichotomy on the kernelizability of \WSP{} under the parameters $k$ for finite $\Gamma$, and $k+m$ for infinite languages $\Gamma$ (subject to the aforementioned computability assumptions). 
This extends kernelization results of Crampton et al.~\cite{CrGuYeJournal}, and represents the first kernelization lower bounds for regular constraint languages. Some results are summarized in Table~\ref{table:overview}.

An interesting open problem is to extend this result beyond regular constraints, e.g., to general user-independent constraints.

\end{document}